\documentclass[12pt]{article}
\usepackage{e-jc-arxiv}

\usepackage{amsmath,amssymb,mathtools}
\usepackage{graphicx}
\usepackage{tikz}
\usetikzlibrary{decorations.markings}

\usepackage[colorlinks=true,citecolor=blue,linkcolor=blue,urlcolor=blue]{hyperref}
\usepackage[nameinlink]{cleveref}

\newcommand{\Oof}{\mathcal{O}}
\newcommand{\CCC}{\mathcal{C}}

\newcommand{\QQQ}{\mathcal{Q}}

\renewcommand{\ker}{\mathrm{ker}}

\newcommand{\cl}{\mathrm{cl}}

\newcommand{\fker}{f_{\ker}}
\newcommand{\fproj}{f_{\mathrm{proj}}}
\newcommand{\fcl}{f_{\cl}}

\newcommand{\fpaths}{f_{\mathrm{pth}}}

\newcommand{\fcore}{f_{\mathrm{core}}}

\newcommand{\N}{\mathbb{N}}

\renewcommand{\phi}{\varphi}
\renewcommand{\epsilon}{\varepsilon}

\newcommand{\minor}{\preccurlyeq}
\newcommand{\dist}{\mathrm{dist}}

\newcommand{\projprof}{\widehat{\mu}}

\newcommand{\abs}[1]{\ensuremath{\left\lvert#1\right\rvert}}

\dateline{Jan 1, 2018}{Jan 2, 2018}{TBD}


\Copyright{Sebastian Siebertz. Released under the CC BY license (International 4.0).}

\title{Reconfiguration on nowhere dense graph classes}

\author{Sebastian Siebertz\thanks{Supported by the National Science Centre of 
Poland via POLONEZ grant agreement UMO-2015/19/P/ST6/03998, 
which has received funding from the European Union's Horizon 2020 research and 
innovation programme (Marie Sk\l odowska-Curie grant agreement No.\ 665778).}\\
\small Institute of Informatics, University of Warsaw, Poland\\[-0.8ex]
\small\tt siebertz@mimuw.edu.pl
}

\begin{document}

\maketitle
\begin{picture}(0,0) \put(393,-385)
{\hbox{\includegraphics[scale=0.25]{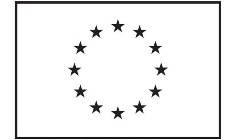}}} \end{picture} 
\vspace{-0.8cm}


\begin{abstract}
\noindent 
Let $\mathcal{Q}$ be a vertex subset problem on graphs. In a 
reconfiguration variant of $\QQQ$ we are given a graph $G$ and 
two feasible solutions $S_s, S_t\subseteq V(G)$ of $\mathcal{Q}$
with $|S_s|=|S_t|=k$. The problem is to determine whether
there exists a sequence $S_1,\ldots,S_n$ of 
feasible solutions, where $S_1=S_s$, $S_n=S_t$, 
$|S_i|\leq k\pm 1$, and  each $S_{i+1}$ results from~$S_i$, 
$1\leq i<n$, by the addition or removal of a single vertex.
We prove that for every nowhere dense class of graphs 
and for every integer $r\geq 1$ there exists a polynomial~$p_r$ such 
that the reconfiguration variants of the \mbox{distance-$r$}
independent set problem and the distance-$r$ dominating
set problem admit kernels of size~$p_r(k)$. If~$k$ is equal to the size
of a minimum distance-$r$ dominating set, then for 
any fixed $\epsilon>0$ we even obtain
a kernel of almost linear size $\mathcal{O}(k^{1+\epsilon})$. 
We then prove that if a class~$\mathcal{C}$ is somewhere dense and 
closed under taking subgraphs, then for some value of $r\geq 1$ the reconfiguration variants of the above problems on $\mathcal{C}$ are 
\mbox{$\mathsf{W}[1]$-hard} (and in particular we cannot expect the existence of kernelization algorithms). 
Hence our results show that the limit of tractability for
the reconfiguration variants of the 
distance-$r$ independent set problem and distance-$r$ dominating set problem on subgraph closed graph classes lies exactly on the boundary
between nowhere denseness and somewhere denseness.

\end{abstract}

\hfill
\pagebreak

\section{Introduction}

In the reconfiguration framework we are not asked to 
find a feasible solution to an optimization problem $\QQQ$, 
but rather to transform a source feasible solution $S_s$
into a more desirable feasible target solution $S_t$
such that each intermediate solution is also feasible. 
This framework allows to model real-world dynamic situations 
in which we need to transform one valid system state
into another and it is crucial that the system keeps running
in all intermediate states. 


The literature focuses mainly on the 
problem of determining the existence of a reconfiguration 
sequence between two given solutions; an even more
difficult problem is to actually find a (possibly minimum-length) 
reconfiguration sequence of solutions.
Typically there are exponentially many feasible solutions
to an instance $I$, and not surprisingly, 
the above problem has been shown to be 
\textsc{PSpace}-complete for the
reconfiguration variants of many important 
\textsc{NP}-complete problems. 

Reconfiguration problems received considerable attention
in recent literature. The studied problems
include \textsc{Vertex Coloring}~\cite{bonamy2013recoloring,
bonsma2009finding,bonsma2014complexity,cereceda2008connectedness,
cereceda2011finding,cereceda2009mixing}, \textsc{List Edge-Coloring}~\cite{ito2012reconfiguration}, \textsc{Vertex Cover}~\cite{mouawad2014vertex,mouawad2014reconfiguration}, \textsc{Independent Set}~\cite{bousquet2017token,hearn2005pspace,ito2011complexity,ito2014fixed,
lokshtanov2017complexity}, \textsc{Clique, Set Cover, Integer Programming, Matching, Spanning Tree, Matroid Ba\-ses}~\cite{ito2011complexity}, \textsc{Satisfiability}~\cite{gopalan2009connectivity,mouawad2015shortest,schwerdtfeger2013computational}, \textsc{Shortest Path}~\cite{bonsma2013complexity, kaminski2011shortest}, \textsc{Subset Sum}~\cite{ito2011approximability}, \textsc{Dominating Set}~\cite{bousquet2017token, haas2014k, haddadan2015complexity, mouawad2017parameterized}, \textsc{Odd Cycle Transversal}, \textsc{Feedback Vertex Set}, and \textsc{Hitting Set}~\cite{mouawad2017parameterized}. We refer to the surveys~\cite{nishimura2017introduction,van2013complexity} and the
thesis~\cite{mouawad2015reconfiguration} for a detailed overview.

A systematic study of the parameterized complexity of reconfiguration problems
was initiated by Mouawad et al.~\cite{mouawad2017parameterized}. 
The authors study mostly graph theoretical vertex subset problems, 
that is, solutions consist of subsets $S\subseteq V(G)$ of the
input graph~$G$. For such problems, one natural parameterization
is the parameter $k$, a bound on the size of feasible solutions, 
another natural parameter is $\ell$, the length of the reconfiguration sequence. 
They proved that \textsc{Feedback Vertex Set}
and \textsc{Bounded Hitting Set} (where the cardinality of
each set from the input is bounded) admit polynomial reconfiguration
kernels (parameterized by~$k$). Concerning lower bounds, 
they proved that reconfiguration of \textsc{Dominating Set} is
$\mathsf{W}[2]$-hard when parameterized by $k+\ell$, as well 
as a general result on reconfigurations of hereditary properties and
their parametric duals, implying $\mathsf{W}[1]$-hardness of
reconfiguration of \textsc{Independent Set, Induced Forest} and
\textsc{Bipartite Subgraph} parameterized by $k+\ell$, and \textsc{Vertex Cover, Feedback Vertex Set,} and \textsc{Odd Cycle Transversal} parameterized by $\ell$. 

In this work we consider the token addition and removal (TAR) 
model of reconfiguration. 
In this model, for a vertex subset
problem $\QQQ$ on graphs, we are given a graph~$G$ and 
two feasible solutions $S_s, S_t\subseteq V(G)$ of $\QQQ$
with $|S_s|=|S_t|=k$. The problem is to determine whether
there exists a sequence $S_1,\ldots,S_n$ of 
feasible solutions, where $S_1=S_s$, $S_n=S_t$, each 
$S_i$ has size $k$ or $k-1$ if $\QQQ$ is a maximization problem
and size $k$ or $k+1$ if $\QQQ$ is a minimization problem, and 
each $S_{i+1}$ results from $S_i$, $1\leq i<n$, 
by adding or removing exactly one vertex. 
\textsc{Independent Set Reconfiguration} is
known to be \textsc{PSpace}-complete on graphs of bounded 
bandwidth~\cite{mouawad2014reconfiguration,wrochna2014reconfiguration}
and W[1]-hard parameterized by $k$ on general graphs~\cite{ito2014parameterized}. 
On the positive side, the problem was shown to be 
fixed-parameter tractable, with parameter $k$, for graphs of 
bounded degree, planar graphs, and graphs excluding $K_{3,d}$ as a 
subgraph, for any constant~$d$~\cite{ito2014fixed,ito2014parameterized}. This result was 
extended by Lokshtanov et al.~\cite{lokshtanov2015reconfiguration} 
to graphs of bounded degeneracy and nowhere dense graphs. 
Lokshtanov~et~al.\ also proved that \textsc{Dominating Set Reconfiguration} is
$\textsf{W}[1]$-hard parameterized by $k+\ell$ on general graphs and fixed-parameter tractable, 
with parameter~$k$, for graphs excluding $K_{d,d}$ 
as a subgraph, for any constant $d$ (in particular on degenerate 
graph classes and nowhere dense classes).

Nowhere dense graph classes, which are also the object of
study in the present paper, are very general classes of uniformly
sparse graphs~\cite{nevsetvril2010first,nevsetvril2011nowhere}. 
Many
familiar classes of sparse graphs, like planar
graphs, graphs of bounded treewidth, graphs of bounded degree, and,
in fact, all classes that exclude a fixed (topological) minor, are nowhere
dense. Notably, classes of bounded average degree or bounded
degeneracy are not necessarily nowhere dense. In an algorithmic context this is
reasonable, as every graph can be turned into a graph of
degeneracy at most~$2$ by subdividing every edge once; however, the
structure of the graph is essentially preserved under this operation.
In our context, a particularly interesting algorithmic result states
that every first-order definable property of graphs can be 
decided in almost linear time on nowhere dense graph 
classes~\cite{grohe2014deciding}. This result implies that 
the reconfiguration variants of many of the above mentioned 
vertex subset problems are fixed-parameter tractable with
respect to parameter $k+\ell$ on every nowhere dense
graph class (the existence of a reconfiguration sequence 
can be expressed with $\Oof(k\cdot \ell)$ quantifiers in 
first-order logic, whenever the property itself can be
defined by a first-order formula), and by the result
of~\cite{grohe2014deciding} can be decided in fixed-parameter
time. 

Nowhere dense graph classes play a special role
for \textsc{Dominating Set} and its more general 
variant \textsc{Distance-$r$ Dominating Set}.
A distance-$r$ dominating set in a graph~$G$ 
(for a fixed integer parameter
$r$) is a set $D\subseteq V(G)$ such that 
every vertex of $G$ is at distance at most $r$ 
to a vertex from $D$. 
\textsc{Distance-$r$ Dominating Set} was shown 
to be fixed-parameter tractable on nowhere dense
classes in~\cite{DawarK09} (this result is again implied
by the more general result of~\cite{grohe2014deciding} which 
was obtained later). 
It was then shown that nowhere dense classes are the
limit of tractability based on sparsity methods, more precisely, 
it was shown in~\cite{drange2016kernelization} that 
if a class $\CCC$ is not nowhere dense and closed under taking
subgraphs, then there is some $r\geq 1$ such that
\textsc{Distance-$r$ Dominating Set} on~$\CCC$ is
$\mathsf{W}[2]$-hard. It was later shown that the problem
admits a polynomial kernel~\cite{siebertz2016polynomial} 
and in fact an almost linear kernel~\cite{eickmeyer2016neighborhood} on nowhere dense classes. 

A kernelization algorithm, or just a kernel, is a polynomial time
algorithm which transforms an input instance $(G,k)$ of a 
parameterized problem to 
an equivalent instance $(G',k')$ such that $|G'|+k'\leq f(k)$
for some function $f$. Hence for a reconfiguration problem
a kernelization algorithm is a polynomial time algorithm which
transforms every input instance $(G,k,S_s,S_t)$ into an 
instance $(G',k', S_s',S_t')$ with $|G'|+k'\leq f(k)$ for some function~$f$ and such that there exists a valid
reconfiguration sequence $S_1=S_s,S_2,\ldots,S_n=S_t$ in $G$
if and only if there exists a valid reconfiguration
sequence $S_1'=S_s',S_2',\ldots, S_m'=S_t'$ in $G'$. 
Every fixed-parameter
tractable problem admits a kernel, however, possibly of exponential 
or worse size. On the reduced instance $(G',k', S_s',S_t')$
one can then run a brute force algorithm to decide whether
the initial instance was a positive instance. 

\subsection{Our results.}

We prove that for every nowhere dense class of graphs 
and for every $r\geq 1$ there exists a polynomial $p_r$ such 
that \textsc{Distance-$r$ Independent Set Reconfiguration} 
and  \textsc{Distance-$r$ Dominating Set Reconfiguration} 
admit kernels of size $p_r(k)$. For \textsc{Distance-$r$ Dominating Set Reconfiguration}, if~$k$ is equal to the size
of a minimum distance-$r$ dominating set of $G$, 
then for any fixed $\epsilon>0$ we even obtain kernels of 
almost linear size $\Oof(k^{1+\epsilon})$. 

%
For \textsc{Distance-$r$ Domination Set Reconfiguration}
there is a technical subtlety that 
prevents us from reducing the input instance $(G,k,I_s,I_t)$ 
to an equivalent instance $(G',k,I_s,I_t)$ such that~$G'$ is a 
subgraph of $G$. Instead, we can kernelize to an 
annotated version of the problem, where only a given subset of vertices of $G'$ needs to be dominated, or we can output
an instance $(G',k,I_s,I_t)$, where $G'$ does not belong
to the class~$\CCC$ under consideration 
(its density parameters are only slightly
larger than those of $G$, though). Formally, in any case, 
we do not reduce
to the same problem, hence we compute only a so-called 
\emph{bi-kernel} for the problem. 
Our results generalize the
earlier mentioned results of Lokshtanov et al.~\cite{lokshtanov2015reconfiguration}, 
who proved that 
\textsc{Independent Set Reconfiguration} (i.e.\ the case $r=1$) 
is fixed-parameter tractable
on every nowhere dense graph class and \textsc{Dominating Set Reconfiguration} 
(i.e.\ the case $r=1$) 
is fixed-parameter tractable if the input graph does not contain large 
complete bipartite graphs (as a subgraph), 
in particular on all nowhere dense graph classes.

Our methods for \textsc{Distance-$r$ Independent Set 
Reconfiguration} generalize
those of Lokshtanov et al.~\cite{lokshtanov2015reconfiguration}
for \textsc{Independent Set Reconfiguration}
to the more general setting of distance-$r$ independence. They 
are strongly based on the equivalence of nowhere denseness 
and uniform quasi-wideness (a notion that will be defined in the next
section) and polynomial bounds for the quasi-wideness 
functions which were recently obtained by Kreutzer et 
al.~\cite{siebertz2016polynomial} and Pilipczuk et al.~\cite{pilipczuk2017wideness}. 

Our methods for \textsc{Distance-$r$ Dominating Set Reconfiguration} combine
the approach of Lokshtanov et al.~\cite{lokshtanov2015reconfiguration} for \textsc{Dominating Set Reconfiguration}
with new methods developed for the kernelization of \textsc{Distance-$r$ Dominating Set} on nowhere dense graph
classes by Eickmeyer et al.~\cite{eickmeyer2016neighborhood}. 

We then prove that if a class~$\CCC$ is somewhere dense and 
closed under taking subgraphs, then for some value of $r\geq 1$ the reconfiguration variants for these problems on $\CCC$ are $\mathsf{W}[1]$-hard (and in particular we cannot expect the
existence of kernelization algorithms). 
Hence our results show that the limit of tractability for
\textsc{Distance-$r$ Independent Set Reconfiguration} 
and \textsc{Distance-$r$ Dominating Set Reconfiguration} 
on subgraph closed graph classes lies exactly 
on the boundary
between nowhere denseness and somewhere denseness.
Our hardness results are rather straightforward generalizations of the
$\mathsf{W}[1]$-hardness proofs known for 
\textsc{Independent Set} and \textsc{Dominating Set} to 
their distance-$r$ variants. 

\section{Preliminaries}

\textbf{Graphs.}
All graphs in this paper 
are finite, undirected and simple.
Our notation is standard, we 
refer to the textbook~\cite{diestel2012graph} for 
more background on graphs.  We write $V(G)$
for the vertex set of a graph~$G$ and $E(G)$ for its
edge set. For $r\in\N$, a graph~$G$ and
$v\in V(G)$ we write $N_r(v)$ for the set of vertices
of $G$ at distance at most~$r$ from~$v$. The 
radius of $G$ is the minimum integer $r$ such that
there is $v\in V(G)$ with $N_r(v)=V(G)$. 

\medskip
\textbf{Independent sets and dominating sets.}
Let $G$ be a graph and $r\in \N$. 
A set $B\subseteq V(G)$ is called {\em{$r$-independent}} in~$G$ if for all
distinct $u,v\in B$ we have $\dist_G(u,v)>r$. The set~$B$ 
is a \emph{distance-$r$ dominating set} in $G$ if 
$N_r(B)=\bigcup_{v\in B}N_r(v)=V(G)$. 

\medskip
\textbf{Minors and subdivisions.}
Let $G$ be a graph and let $r\in \N$. 
A graph~$H$ with vertex set
$\{v_1,\ldots, v_n\}$ is a \emph{depth-$r$ minor} of~$G$, written
$H\minor_r G$, if there are connected and pairwise vertex disjoint
subgraphs $H_1,\ldots, H_n\subseteq G$, each of radius at most $r$,
such that if $v_iv_j\in E(H)$, then there are $w_i\in V(H_i)$ and
$w_j\in V(H_j)$ with $w_iw_j\in E(G)$.

An \emph{$r$-subdivision} of $H$ is obtained by replacing edges of $H$
by internally vertex disjoint paths of length (exactly) $r$.  We write 
$H_r$ for the $r$-subdivision of $H$.

\medskip
\textbf{Nowhere denseness.}
  A class $\CCC$ of graphs is \emph{nowhere dense} if there exists a
  function $t\colon \N\rightarrow \N$ such that $K_{t(r)}\not\minor_r G$ 
  for all $r\in\N$ and for all $G\in \CCC$. Otherwise, 
  $\CCC$ is called \emph{somewhere dense}.


\medskip
\textbf{Uniform quasi-wideness.}
  A class $\CCC$ of graphs is called \emph{uniformly quasi-wide} if there are
  functions $N\colon \N\times\N\rightarrow \N$ and $s:\N\rightarrow \N$ such
  that for all $r,m\in \N$ and all subsets $A\subseteq V(G)$ for
  $G\in \CCC$ of size $\abs{A}\geq N(r,m)$ there is a set
  $S\subseteq V(G)$ of size $\abs{S}\leq s(r)$ and a set
  $B\subseteq A\setminus S$ of size $\abs{B}\geq m$ which is $r$-independent in
  $G-S$.  

\smallskip
It was shown by Ne\v{s}et\v{r}il and Ossona de
Mendez~\cite{nevsetvril2011nowhere} that a class $\CCC$ of graphs is
nowhere dense if and only if it is uniformly quasi-wide. Quasi-wideness
is a very useful property for distance-$r$ domination, as large $2r$-independent
sets are natural obstructions for small distance-$r$ dominating sets. 
For us it will be important that the function~$N$ can be assumed 
to be polynomial in~$m$ (the degree of the polynomial may depend on~$r$) and that the sets~$B$ and~$S$ can be efficiently
computed. Polynomial bounds were first obtainend 
by Kreutzer et al.~\cite{siebertz2016polynomial}, we refer 
to the improved bounds of Pilipczuk et al.~\cite{pilipczuk2017wideness}.

\begin{lemma}[Pilipczuk et al.~\cite{pilipczuk2017wideness}]\label{thm:uqw}
Let $\CCC$ be a nowhere dense class of graphs. 
For all $r\in \N$ there is a polynomial  $N_r\colon \N\to \N$ 
and a constant $t_r\in \N$, such that the following holds.
Let $G\in \CCC$ and
let $A\subseteq V(G)$ be a vertex subset of size at least $N_r(m)$, for a given $m$.
Then there exists a set $S\subseteq V(G)$ of size $|S|\leq t_r$ and a set $B\subseteq A\setminus S$ 
of size $|B|\geq m$ which is $r$-independent in $G-S$.
Moreover, given~$G$ and $A$, such sets $S$ and $B$ can be computed in time $\Oof(|A|\cdot |E(G)|)$. 
\end{lemma}

We remark that the $\Oof$-notation in the above lemma 
hides constant factors depending on~$r$ (which is considered
fixed) and the class $\CCC$.

\medskip
\textbf{A-avoiding paths.}
Let $G$ be a graph and let $A\subseteq V(G)$ be a subset of vertices. For vertices $v\in A$ and $u\in V(G)$, a path $P$ connecting $u$ and 
$v$ is called {\em{$A$-avoiding}}
if all its vertices apart from $u$ and $v$ do not belong to $A$.

\medskip
\textbf{Projection profiles.}
Let $G$ be a graph, $A\subseteq V(G)$ and $r\in \N$. The {\em{$r$-projection}} of a vertex $u\in V(G)$ onto~$A$, denoted $M^G_r(u,A)$ 
is the set of all vertices $v\in A$ that
can be connected to $u$ by an $A$-avoiding path of length at most $r$. The {\em{$r$-projection profile}} of a vertex $u\in V(G)$ on $A$ is the function $\rho^G_r[u,A]$ mapping vertices of
$A$ to $\{0,1,\ldots,r,\infty\}$, defined as follows: for every $v\in A$, the value $\rho^G_r[u,A](v)$ is the length of a shortest $A$-avoiding path connecting $u$ and~$v$ and $\infty$ in case this length
is larger than~$r$. 
We define 
\[\projprof_r(G,A)=|\{\rho_r^G[u,A]\colon u\in V(G)\}|\]
to be the number of different $r$-projection profiles realized on $A$. 
For $u,v\in V(G)$ we define \[u\cong_{A,r} v \quad\Longleftrightarrow\quad \rho_r^G[u,A]=\rho_r^G[v,A].\]

\begin{lemma}[Eickmeyer et al.~\cite{eickmeyer2016neighborhood}]\label{lem:projection-complexity}
  Let $\CCC$ be a nowhere dense class of graphs. Then there is 
  a function $\fproj(r,\epsilon)$ such that for every $r\in \N$, 
  $\epsilon>0$, graph $G\in \CCC$, and vertex subset $A\subseteq V(G)$, 
  it holds that $\projprof_r(G,A)\leq \fproj(r,\epsilon)\cdot |A|^{1+\epsilon}$.
\end{lemma}

We remark that in \cite{eickmeyer2016neighborhood} 
$r$-projections onto $A$ are defined only for vertices which do 
not lie inside $A$ themselves. This does not affect the statement of 
\Cref{lem:projection-complexity}, as this change of definition
accounts only for a term $|A|$, which can be absorbed in the
function $\fproj$. 

\medskip
\textbf{Parameterized complexity.}
A problem is fixed-parameter tractable on a class~$\CCC$ of
graphs with respect to a parameter $k$, if there is an
algorithm deciding whether a graph $G\in \CCC$ admits a solution of
size $k$ in time $f(k)\cdot \abs{V(G)}^c$, for a computable function
$f$ and constant~$c$. A kernelization algorithm is a polynomial time 
algorithm which reduces the input instance to a 
sub-instance of size bounded in the
parameter only (independently of the input graph size). Every fixed-parameter
tractable problem admits a kernel, however, possibly of exponential or worse size. 
For efficient algorithms it is therefore most desirable to obtain polynomial, 
or optimally even linear, kernels. 
The $\mathsf{W}$-hierarchy is a collection of parameterized complexity classes
$\mathsf{FPT}\subseteq \mathsf{W}[1]\subseteq \mathsf{W}[2]\subseteq \ldots$. 
The assumption $\mathsf{FPT}\subsetneq \mathsf{W}[1]$ can be seen as the 
analogue of the conjecture that $\mathsf{P}\subsetneq \mathsf{NP}$. Therefore, 
showing hardness in the parameterized setting is usually accomplished by establishing
an $\mathit{fpt}$-reduction to a $\mathsf{W}[1]$-hard problem. We refer to 
the textbooks \cite{cygan2015parameterized,downey2013fundamentals,flum2006parameterized} 
for extensive background on parameterized complexity. 

\medskip
\textbf{Reconfiguration.}
The token addition and removal variant of the 
\textsc{Distance-$r$ Independent Set Reconfiguration} problem ($r$-ISR) is defined as follows. On input $(G,k,I_s,I_t)$, where $G$ is a graph, 
$k\in \N$, and $I_s,I_t$ are distance-$r$ independent sets
in $G$ of size $k$, the problem is to determine whether there exists a sequence $I_s=I_1,\ldots, I_\ell=I_t$ such that for all $1\leq j\leq \ell$
\begin{enumerate}
\vspace{-1mm}
\item $I_j$ is a distance-$r$ independent set in $G$, \vspace{-1mm}
\item $|\Delta(I_j,I_{j+1})|=|(I_j\setminus I_{j+1})\cup (I_{j+1}\setminus I_j)|=1$, and\vspace{-1mm}
\item $k-1\leq |I_j|\leq k$.
\end{enumerate}

%

The \textsc{Distance-$r$ Dominating Set Reconfiguration} problem ($r$-DSR) is defined analogously, we only demand that in the fourth item we 
have $k\leq |D_j|\leq k+1$ for the appearing distance-$r$ dominating
sets $D_j$, $1\leq j\leq \ell$. 
We obtain positive results also for the variants where for $r$-ISR we get as input two integer parameters
$k,k'$ and we replace the 
fourth condition by $k\leq |I_j|\leq k'$ for all $1\leq j\leq \ell$. For $r$-DSR we may
remove the condition on a lower bound completely, that is, in the forth condition 
demand only that $|D_j|\leq k+1$ for all $i\leq j\leq \ell$.

\section{Distance-r independent set reconfiguration}

\subsection{Polynomial kernel}

Our approach for kernelization of $r$-ISR is similar to that of
Lokshtanov et al.~\cite{lokshtanov2015reconfiguration}. 
We iteratively remove \emph{irrelevant vertices} from 
the input instance, until this is no longer possible, in which case
the resulting instance will be small. 

\medskip
\textbf{Irrelevant vertices.}
Let $(G,I_s,I_t,k)$ be an instance of \textsc{Distance-$r$ Independent
Set Reconfiguration}. A vertex $v\in V(G)\setminus 
(I_s\cup I_t)$ is called
an \emph{irrelevant vertex} if $(G, k, I_s, I_t)$ is a positive 
instance if and only if $(G-v, k,I_s, I_t)$
is a positive instance.

\medskip
The following lemma shows that given $G$ is large, we can efficiently
find an irrelevant vertex.

\begin{lemma}\label{lem:irrelevantvertex}
Let $(G,k,I_s,I_t)$ for $G\in\CCC$ be an instance of \textsc{Distance-$r$ Independent Set Reconfiguration}, $M\coloneqq I_s\cup I_t$  and let 
$R\coloneqq V(G)\setminus M$. Let $S\subseteq V(G)$ and 
$B\subseteq R\setminus S$ such that $B$ is $2r$-independent in $G-S$. 
Furthermore, assume that all vertices of $B$ have the same $r$-projection
profile to $S$, i.e., $\rho_{r}^G[u,S]=\rho_r^G[v,S]$. 
If $|B|\geq 2k$, then any $v\in B$ is an irrelevant vertex. 
\end{lemma}
\begin{proof}
Let $v\in B$ and enumerate $2k-1$ vertices of $B\setminus\{v\}$
as $w_1,\ldots, w_{2k-1}$. We aim to show that~$v$ is an 
irrelevant vertex. Observe that since $B\subseteq R\setminus S$
the set $\{v,w_1,\ldots, w_{2k-1}\}$ is disjoint from the set $M\cup S$. 

Consider a reconfiguration sequence $I_s=I_1,I_2,\ldots, I_t=I_\ell$ from $I_s$ to $I_t$ in $G$ with a minimum number of occurrences 
of $v$. We want to prove that $v$ does not occur at all in the 
sequence, as this proves that $v$ is irrelevant. Towards a 
contradiction assume that $v$ does occur in the sequence
and let $p$, $1<p<\ell$, be the first index at which $v$ appears in $I_p$ (that is, $v\in I_p$ and $v\not\in I_i$ for all 
$i<p$).
Let $q+1, p<q+1\leq \ell$ be the first index after~$p$ at 
which~$v$ is 
removed (that is, $v\in I_p,\ldots, I_q$ and $v\not\in I_{q+1}$). 
We will modify the sub-sequence $I_p,\ldots, I_q$ such that it does not use $v$, contradicting our choice of a reconfiguration 
sequence with a minimum number of occurrences of $v$. 
Fix some $j$, $p\leq j\leq q$, and let $I=I_j\setminus\{v\}$ and 
$I'=I_{j+2}\setminus \{v\}$. 

\setcounter{theorem}{0}
\begin{claim}\label{cl:cl1}
If there is
$z\in I$ with \mbox{$\dist_{G}(w_i,z)\leq r$} for 
some $w_i$, then $\dist_G(w_\ell,z)>r$ for all $\ell\neq i$. 
\end{claim}
\begin{proof}\renewcommand{\qedsymbol}{$\lrcorner$}
Assume towards a contradiction that there is $\ell\neq i$ 
with $\dist_G(w_\ell,z)\leq r$. Let $P_i$ be a shortest 
path (of length at most $r$) between $w_i$ and $z$ 
and let $P_\ell$ a shortest path (of length at most $r$) 
between $w_\ell$ and $z$. 
As $B$ is $2r$-independent in $G-S$ there exists a vertex
$s\in S$ with $s\in V(P_i)$ or $s\in V(P_\ell)$. By symmetry
we may assume that $s\in V(P_i)$ and assume that among
all vertices of $S$ which lie on $P_i$, the vertex $s$ is the 
one which is closest to $w_i$. 
Then we have $\dist_G(w_i,z)=\dist_G(w_i,s)+\dist_G(s,z)$. 
Now we have \mbox{$\rho_r^G[v,S]=
\rho_r^G[w_i,S]$}, hence $\dist_G(v,s)=\dist_G(w_i,s)$. 
This implies 
\[\dist_G(v,z)\leq \dist_G(v,s)+\dist_G(s,z)= \dist_G(w_i,s)+\dist_G(s,z)=\dist_G(w_i,z)\leq r,\]
contradicting that $I_j$ is a distance-$r$ independent set. 
\end{proof}

\begin{claim}\label{cl:cl3}
There exists $w\in\{w_1,\ldots, w_{2k-1}\}$ with $(I\cup I')
\cap N_r(w)=\emptyset$. 
\end{claim}
\begin{proof}\renewcommand{\qedsymbol}{$\lrcorner$}
For $z\in I$, if there is $w_i\in \{w_1,\ldots, w_{2k-1}\}$
with $z\in N_r(w_i)$, i.e., \mbox{$\dist_G(z,w_i)\leq r$}, then 
by \Cref{cl:cl1} we have $\dist_G(w_\ell,z)>r$ for all $\ell\neq i$. 
As the set~$I$ contains only $k-1$ elements we 
conclude that there are $k$ elements $u_1,\ldots, u_k$ in 
$\{w_1,\ldots, w_{2k-1}\}$ with $I\cap N_r(u_i)=\emptyset$
for all $1\leq i\leq k$. We apply the same reasoning to the
set $I'$ and $\{u_1,\ldots, u_k\}$
(note that the claims are also applicable to $I'$, 
as $j$ is chosen arbitrary). This leaves us with an element
$w\in\{w_1,\ldots, w_{2k-1}\}$ with $(I\cup I')
\cap N_r(w)=\emptyset$. 
\end{proof}

As $j$ was chosen arbitrary, we conclude that for every $j$
there exists an element $w^j\in \{w_1,\ldots, w_{2k-1}\}$ such that 
$(I_j\setminus\{v\})\cup\{w^j\}$ and 
$(I_{j+2}\setminus\{v\})\cup\{w^j\}$ are distance-$r$ independent 
sets of the same size as $I_j$. Note that we have $|I_p|=|I_q|=k$, 
as $v$ was introduced at $I_p$ and removed at $I_{q+1}$.
Hence, if we have $j=p+2x$ for some $x\in \N$, then 
$I_{j+1}\subseteq I_j,I_{j+2}$, hence also $(I_{j+1}\setminus\{v\})
\cup\{w^j\}$ is a distance-$r$ independent sets of size $k-1$. 

\smallskip
We now modify the sequence $I_p,\ldots, I_q$ as follows. 
For each $j=p+2x$ for some $x\in\N$ such that $p\leq j< q$, 
we replace the subsequence $I_j\rightarrow I_{j+1}$ of the reconfiguration
sequence by the sequence 
\[(I_j\setminus\{v\})\cup\{w^j\} \rightarrow (I_{j+1}\setminus \{v\})\cup \{w^j\}\rightarrow (I_{j+2}\setminus \{v\})\cup \{w^j\}\rightarrow (I_{j+2}\setminus \{v\})\]
and we replace $I_q$ by $(I_q\setminus\{v\})\cup \{w^q\}$. 

\smallskip
By our above argument, each of the intermediate 
configurations is a distance-$r$ independent set of size $k$ or $k-1$. Furthermore,
the transition $I_{p-1}, J_p$ is valid, so is every intermediate
transition and the transition $J_q, I_{q+1}$. This finishes the
proof of the lemma. 
\end{proof}

\setcounter{theorem}{3}
\begin{theorem}\label{thm:reconfig}
Let $\CCC$ be a nowhere dense class of graphs and let $r\in \N$. 
Let $(G,k,I_s,I_t)$ be an instance of \textsc{Distance-$r$ Independent
Set Reconfiguration}, where $G\in \CCC$. Then we can 
compute in polynomial time a subgraph $G'\subseteq G$ with 
$I_s,I_k\subseteq V(G')$ such that $(G,k,I_s,I_t)$ is a positive
instance if and only if $(G',k,I_s,I_t)$ is a positive instance
and $G'$ has order polynomial in $k$. 
\end{theorem}
\begin{proof}
Let $N=N_{2r}:\N\rightarrow \N$ be the function 
and $t=t_{2r}\in \N$ be the constant 
describing~$\CCC$ as uniformly quasi-wide (for 
parameter $2r$) as defined in \Cref{thm:uqw}.
Let $M\coloneqq I_s\mathop{\cup} I_t$ of size $2k$
and let $R\coloneqq V(G)\setminus M$. 

If $|R|\geq N(2k\cdot 
(r+2)^t)$, according to \Cref{thm:uqw} we can compute
in polynomial time a set $S\subseteq V(G)$ of size $|S|\leq t$
and a set $B\subseteq R\setminus S$ of size $|B|\geq 
2k(r+2)^t$ which is $2r$-independent in $G-S$. 
We classify the elements of $B$ with respect to their 
$r$-projections to the set $S$. The corresponding 
equivalence relation $\cong_{S,r}$ on $B$ 
has at most $(r+2)^t$ equivalence classes, 
as $|S|\leq t$ and $\rho_r^G[u,S]$ is a mapping from $S$ to 
$\{0,1,\ldots, r,\infty\}$. 
Since $|B|\geq 2k(r+2)^t$, 
we know that at least one equivalence class 
contains at least $2k$ vertices of $B$. 
We apply \Cref{lem:irrelevantvertex} to this
equivalence class to find an irrelevant vertex $v$. 
We remove $v$ from the graph and iterate this procedure
until $|R|<N(2k\cdot (r+2)^t)$. 
In this case the resulting graph has size at most 
$N(2k\cdot (r+2)^t)+2k$, 
which is polynomial in $k$ for each fixed value of~$r$. 
\end{proof}

It is easy to see that we can carry out the same proof for the 
reconfiguration variant where 
we get as input two integer parameters $k,k'$ and we replace the 
fourth condition by $k\leq |I_j|\leq k'$ for all $1\leq j\leq \ell$. The kernel
will have size polynomial in $k'$. 

\smallskip
We remark that the kernel does possibly not preserve the 
length of a shortest reconfiguration sequence. It remains an 
interesting open question to compute a kernel with this preservation
property. 

\subsection{Lower bounds}

Recall that for a graph $G$ and 
$s\in\N$, $G_s$ denotes the $s$-subdivision of $G$. 
Our hardness result is based on the following observation
by Ne\v{s}et\v{r}il and Ossona de Mendez.

\begin{lemma}[Ne\v{s}et\v{r}il and Ossona de Mendez~\cite{nevsetvril2011nowhere}, see also
\cite{drange2016kernelization}]\label{lem:subdiv}
Let $\CCC$ be somewhere dense and closed under taking
subgraphs. Then there is $s\in \N$ such that for all graphs $G$
we have $G_s\in\CCC$. 
\end{lemma}

Furthermore, we use that \textsc{Independent Set Reconfiguration}, 
i.e., the case $r=1$ is hard. 

\begin{lemma}[Ito et al.~\cite{ito2014parameterized}]\label{lem:is-whard}
\textsc{Independent Set Reconfiguration} is $\mathsf{W}[1]$-hard. 
\end{lemma}

\begin{theorem}\label{thm:isw1}
Let $\CCC$ be somewhere dense and closed under taking
subgraphs. Then there is $r\in \N$ such that 
\textsc{Distance-$r$ Independent Set 
Reconfiguration} is $\mathsf{W}[1]$-hard on~$\CCC$. 
\end{theorem}
\begin{proof}
According to \Cref{lem:subdiv}, there is $s\in \N$ such that
for all graphs $G$ we have $G_s\in \CCC$. We reduce 
$1$-ISR to $(4s-1)$-ISR on $\CCC$ by establishing the following. 
For each graph $G$ there exists a polynomial time computable graph 
$H\in\CCC$ such that $V(G)\subseteq V(H)$ and such that
\begin{enumerate}
\item every independent set $I$ in $G$ is a $(4s-1)$-independent
set in $H$ and 
\item every $(4s-1)$-independent set $I$ of size at least $2$ in 
$H$ consists only of vertices which are also vertices of $G$ and
$I$ is an independent set in $G$. 
\end{enumerate}

Note that we may assume that the parameter $k$ is always at least
$2$. The above properties guarantee that every reconfiguration 
sequence $I_1,\ldots, I_\ell$ of independent sets in~$G$
corresponds uniquely to a reconfiguration sequence of distance-$(4s-1)$
independent sets in $H$ and vice versa. Hence, we can conclude 
the statement of the theorem for $r=4s-1$ by
applying \Cref{lem:is-whard}. Note that the reduction also establishes
$\mathsf{W}[1]$-hardness of \textsc{Distance-$r$ Independent Set}
on somewhere dense graph classes which are closed under taking 
subgraphs.

Let $G$ be a graph with vertex set $\{v_1,\ldots,v_n\}$ and
edge set $\{e_1,\ldots,e_m\}$. We remark that the hardness result
 of \Cref{lem:is-whard} works also if we assume that all input graphs do not
 have isolated vertices, so we may assume that $G$ does not contain isolated vertices. We define the new graph
$J$ with vertex set \[\{v_1,\ldots,v_n,e_1,\ldots, e_m, w\}\] 
and edge set \[\{ve:v\in V(G), e\in E(G), v\in e\}\cup \{ew :e\in E(G)\}.\] 
We claim that $H=J_s$ satisfies the above claimed properties. 

\bigskip
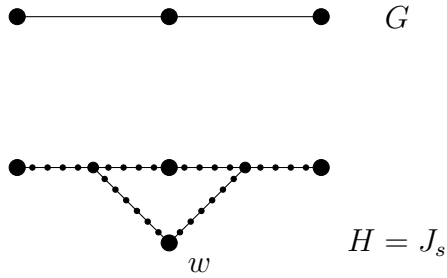
\begin{figure}[h]
\begin{center}
\begin{tikzpicture}[decoration={markings,
  mark=between positions 0 and 1 step 0.2cm
  with { \draw [fill] (0,0) circle [radius=1pt];}}]
\filldraw[fill=black] (0,0) circle (3pt);
\filldraw[fill=black] (2,0) circle (3pt);
\filldraw[fill=black] (4,0) circle (3pt);
\node at (5,0) {$G$};
\draw[-] (0,0) -- (4,0);

\filldraw[fill=black] (0,-2) circle (3pt);
\filldraw[fill=black] (2,-2) circle (3pt);
\filldraw[fill=black] (4,-2) circle (3pt);
\path[postaction={decorate}]  (0,-2) to (4,-2);
\draw[-] (0,-2) to (4,-2);

\filldraw[fill=black] (1,-2) circle (2pt);
\filldraw[fill=black] (3,-2) circle (2pt);

\filldraw[fill=black] (2,-3) circle (3pt);
\path[postaction={decorate}]  (2,-3) to (1,-2);
\path[postaction={decorate}]  (2,-3) to (3,-2);
\draw[-] (2,-3) to (3,-2);
\draw[-] (2,-3) to (1,-2);
\node at (2.4,-3.3) {$w$};
\node at (5,-3) {$H=J_s$};

\end{tikzpicture}
\end{center}
\caption{A graph $G$ and the constructed graph $H\in \CCC$. Vertices
at distance $1$ in $G$ have distance $2s$ in $H$, while vertices at 
distance $2$ have distance $4s$. All vertices of $H$ which do not
correspond to vertices of $G$ have distance at most $4s-1$.}
\end{figure}

Let $I$ be a distance-$1$ independent set in $G$. By construction
of $J$, if $u,v\in V(G)$ are adjacent in $G$, then they have 
distance $2s$ in $J_s$, otherwise they have distance $4s$ in 
$J_s$ (via the vertex~$w$). Hence, $I$ is a distance-$(4s-1)$
independent set in $J_s$. 

Conversely, let $I$ be a distance-$(4s-1)$ independent set
in $J_s$ of size at least $2$. First observe that $I$ consists
only of vertices which are also vertices of $G$. All other
vertices have mutual distance at most $4s-1$ via the vertex $w$. 
As seen above, the elements of $I$ have distance $4s$ in~$J_s$
and therefore distance at least $2$ in $G$, that is, $I$ is
an independent set in~$G$. This finishes the proof. 
\end{proof}

\section{Distance-r dominating set reconfiguration}

\subsection{Polynomial kernel}

The kernelization for \textsc{Distance-$r$ Dominating Set Reconfiguration} strongly depends on the following
notion of a domination core which was (in different variants)
also used in the earlier kernelization results for distance-$r$
dominating sets~\cite{DawarK09,drange2016kernelization,eickmeyer2016neighborhood}.

\medskip
\textbf{Domination core.}
Let $G$ be a graph and $k,r\in\N$. A set $Z\subseteq V(G)$ is called a \emph{$(k,r)$-domination core} if every set $D$ of size at most $k$ that $r$-dominates $Z$ also $r$-dominates~$G$. 

\smallskip
Clearly, $V(G)$ is a $(k,r)$-domination core. Hence, 
starting with $Z=V(G)$, using the next lemma, we can 
gradually remove vertices from $Z$ while maintaining the 
invariant that $Z$ is a $(k,r)$-domination core. 
The proof of the lemma is the same as the proof of
Lemma 11 in \cite{DawarK09} and Lemma 4.1 in \cite{siebertz2016polynomial}, we just use the better bounds from \Cref{thm:uqw}.

\begin{lemma}\label{lem:findredundantvertex}
Let $\CCC$ be a nowhere dense class of graphs and 
let $k,r\in\N$. Let $N=N_{2r}$ and $t=t_{2r}$ 
be the functions characterizing $\CCC$ as uniformly quasi-wide according to \Cref{thm:uqw} with parameter $2r$. 
  There is an algorithm
  that, given a graph $G\in\CCC$, $k\in\N$ and $Z\subseteq V(G)$ with 
  $\abs{Z}>N\bigl((k+2)(2r+1)^t\bigr)\eqqcolon \ell$ runs in time $\Oof(t\cdot \ell\cdot |E(G)|)$, and  
  returns a vertex $w\in Z$ such that for 
  any set $X\subseteq V(G)$
  with $\abs{X}\leq k$,
  it holds that $X$ is an $r$-dominating set of $Z$ if, and only if, 
  $X$ is an $r$-dominating set of
    $Z\setminus \{w\}$.
\end{lemma}

We iteratively apply \Cref{lem:findredundantvertex}
for at most $n$ times, until this is no longer possible. This yields the
following lemma. 

\begin{lemma}\label{lem:findcore}
Let $\CCC$ be a nowhere dense class of graphs and 
let $k,r\in\N$. There exists a polynomial~$q_r$ and a
polynomial time algorithm
  that, given a graph $G\in\CCC$ and $k\in\N$
either correctly concludes that $G$ cannot be $r$-dominated by a 
set of at most $k$ vertices, or finds a $(k,r)$-domination core $Z\subseteq V(G)$ of $G$ of size at most $q_r(k)$. 
\end{lemma}


We now define the annotated problem \textsc{Distance-$r$ $Z$-Dominating Set}, $r$-ZDS, as the problem to find on input
$(G,Z,k)$ a set $D$ with $Z\subseteq N_r(D)$. Such a set $D$
is called a \emph{$(Z,r)$-dominator}. By definition, if $Z$ is a $(k,r)$-domination 
core, then every $(Z,r)$-dominator of size 
at most $k$ corresponds to a distance-$r$ dominating set of
$G$. On the other hand, every distance-$r$ dominating set
of $G$ in particular dominates every subset $Z\subseteq V(G)$. 
%
We define the reconfiguration variant of the problem, 
$r$-ZDSR, in the obvious way. 


\begin{theorem}\label{thm:reconfig-ds}
Let $\CCC$ be a nowhere dense class of graphs and let $r\in \N$. 
Let $(G,k,D_s,D_t)$ be an instance of $r$-DSR, where $G\in \CCC$. 
We can 
compute in polynomial time a subgraph $G'\subseteq G$ with 
$D_s,D_k\subseteq G'$ and $Z\subseteq V(G')$ 
such that $(G,k,D_s,D_t)$ is a positive
instance of $r$-DSR if and only if $((G',Z),k,D_s,D_t)$ is a positive instance
of $r$-ZDSR and $G'$ has order polynomial in $k$. 
\end{theorem}
\begin{proof}
We compute a $(k,r)$-domination core $Z\subseteq V(G)$
of size at most $q_r(k)$ using \Cref{lem:findcore}. 
Let $\epsilon>0$ and let $\fproj(r,\epsilon)$ be the 
function from \Cref{lem:projection-complexity}. 
According to the lemma there are at most 
$\fproj(r,\epsilon)\cdot |Z|^{1+\epsilon}$ different $r$-projections
to $Z$. Recall that
$u\cong_{Z,r} v\Leftrightarrow \rho_r^G[u,Z]=\rho_r^G[v,Z]$.
Now for each projection class $\kappa$ we choose a representative~$v_\kappa$ from that class. 

\smallskip
\setcounter{theorem}{0}
\begin{claim}
For all $u,v\in V(G)$, $\rho_r^G[u,Z]=\rho_r^G[v,Z]$ implies $N_r^G(u)\cap Z=N_r^G(v)\cap Z$.
\end{claim}
\begin{proof}\renewcommand{\qedsymbol}{$\lrcorner$}
Let $z\in N_r^G(u)\cap Z$ and let $P$ be a shortest path between $u$ and $z$. 
If $P$ is $Z$-avoiding we conclude from $\rho_r^G[u,Z]=\rho_r^G[v,Z]$ that there 
exists also a $Z$-avoiding path~$P'$ of the same length as $P$ between $v$ and $z$, 
which implies $z\in N_r^G(v)\cap Z$. Otherwise, let $z'$ be the vertex of $Z$ on $P$ which is closest to $u$
and let $Q$ be the initial part of $P$ between $u$ and~$z'$. Note that $Q$ is a shortest
path between $u$ and $z'$. By the same argument as above, we find a $Z$-avoiding 
path~$Q'$ of the same length as $Q$ between $v$ and $z'$. By replacing~$Q$ in~$P$ 
by~$Q'$ we obtain a path of the same length at $P$ between $v$ and $z$, which again
proves $z\in N_r^G(v)\cap Z$. 
\end{proof}

We now construct $G'$ such that it contains $D_s, D_t$, the 
set $Z$, all the representatives $v_\kappa$ and furthermore
a small set $T$ of vertices such that $N_r^G(v_\kappa)\cap Z=N_r^{G'}(v_\kappa)\cap Z$. 
The set~$T$ is constructed as follows. For each $v\in V(G)$, let $T_v$
be a breadth-first search tree with root~$v$ of depth $r$ which has 
elements of~$Z$ as its leaves. Clearly, $\dist_G(v,z)=\dist_{T_{v}}(v,z)$
for all $z\in N_r(v)\cap Z$. Let $T$ be the set $\bigcup_{\kappa}V(T_{v_\kappa})\cup \bigcup_{v\in D_s\cup D_t}V(T_v)$.
Hence for each $v_\kappa$ we have $N_r^{G'}(v_\kappa)\cap Z=N_r^{G}(v_\kappa)\cap Z$,
which immediately implies the next claim. 

\bigskip 

\begin{claim}
Let $D'$ be a $(Z,r)$-dominator in $G'$ which contains only representative
vertices~$v_\kappa$. Then $D'$ is also a $(Z,r)$-dominator in $G$. \hfill$\lrcorner$
\end{claim}

\medskip
We will always find $(Z,r)$-dominators of the above form. 

\medskip
\begin{claim}
Let $v\in V(G')$. Then there is $v_\kappa\in G'$ such that 
$N_r^{G'}(v)\cap Z\subseteq N_r^{G'}(v_\kappa)\cap Z$. 
\end{claim}
\begin{proof}\renewcommand{\qedsymbol}{$\lrcorner$}
Let $\kappa$ be the equivalence class of $v$ in 
the relation $\cong_{Z,r}$. Then 
$N_r^{G'}(v)\cap Z\subseteq N_r^{G'}(v_\kappa)\cap Z$. 
\end{proof}

\bigskip
Conversely, $(Z,r)$-dominators in $G$ can be translated to 
$(Z,r)$-dominators in $G'$. 
\begin{claim}
Let $D$ be a distance-$r$ dominating set in $G$. Then $D'=
\{v_\kappa : v\in D, v\cong_{Z,r} v_\kappa\}$ 
is a 
$(Z,r)$-dominator in $G'$.
\end{claim}
\begin{proof}\renewcommand{\qedsymbol}{$\lrcorner$}
As~$v_\kappa$ is chosen so that 
$\rho_r^G[v,Z]=\rho_r^G[v_\kappa,Z]$, by Claim 1
it holds that $N_r^G(v)\cap Z=N_r^G(v_\kappa)\cap Z$. Hence 
$Z\subseteq N_r^G(D)$ and $N_r^{G'}(v_\kappa)\cap Z=N_r^{G}(v_\kappa)\cap Z$
 implies that also $Z\subseteq N_r^{G'}(D')$. 
\end{proof}

We can now prove that the instance $((G',Z),k,D_s,D_t)$ of $r$-ZDSR is equivalent
to the instance $(G,k,D_s,D_t)$. If $D_1,\ldots, D_n$ is a valid reconfiguration sequence in 
$G$, then according to Claim 4 the corresponding sequence 
$D_1',\ldots, D_n'$ is also a valid reconfiguration sequence of 
$(Z,r)$-dominators in $G'$. 

Conversely, Let $D_1',\ldots,D_n'$ be a reconfiguration sequence of 
$(Z,r)$-dominators in $G'$. We first modify $D_i'$ such that it uses only
representative vertices $v_\kappa$, using Claim 3. Now according to 
Claim 2, $D_i'$ is also a $(Z,r)$-dominator in $G$. By definition of 
a $(k,r)$-domination core, $D_i'$ is a distance-$r$ dominating set in $G$. 

It remains to estimate the size of $G'$. According to \Cref{lem:findcore}, 
$Z$ has polynomial size at most $q_r(k)$. According to \Cref{lem:projection-complexity}
there are at most $\fproj(r,\epsilon)\cdot |Z|^{1+\epsilon}$ projection 
classes, hence we add at most so many vertices $v_\kappa$ to 
$G'$. Furthermore, each spanning tree~$T_{v_\kappa}$ has 
order at most $r\cdot |Z|$. Together with the $2k$ spanning
trees we add for $D_s$ and~$D_t$, we have
$|V(G')|\leq (\fproj(r,\epsilon)+2k)\cdot q_r(k)^{2+\epsilon}\cdot r$, which 
is polynomial for every fixed value of~$r$ and~$\epsilon$.
\end{proof}

The annoying fact that we reduce to an annotated version
of the problem can be dealt with by introducing a simple
gadget to $G'$. The same problem occurred also in the 
kernelization algorithms for \textsc{Distance-$r$ Dominating
Set} on bounded expansion and nowhere dense graph 
classes~\cite{drange2016kernelization, eickmeyer2016neighborhood}. 
We refer to these papers for the (very simple) 
details. 

\bigskip
We can find much smaller domination cores if we make a
further assumption on the dominating set size. The following
definition was first given in~\cite{drange2016kernelization}
and is also the basis for the kernelization of 
\textsc{Distance-$r$ Dominating Set} in~\cite{eickmeyer2016neighborhood}.

\medskip
\textbf{Minimum domination core.}
Let $G$ be a graph and $r\geq 1$. A set $Z\subseteq V(G)$ is a \emph{distance-$r$ domination core} 
if every set $D$ of minimum size that $r$-dominates $Z$ also $r$-dominates $G$.

\smallskip
The little change in the definition makes a large difference for the sizes of 
the respective cores, as the next lemma shows. 

\setcounter{theorem}{10}
\begin{lemma}[Eickmeyer et al.~\cite{eickmeyer2016neighborhood}]\label{lem:core2}
Let $\CCC$ be a nowhere dense class of graphs. There exists a function $\fcore(r,\epsilon)$ and a polynomial-time algorithm that, given a graph $G\in \CCC$, integer $k\in \N$ and $\epsilon>0$, 
either correctly concludes that $G$ cannot be $r$-dominated by $k$ vertices, or finds a distance-$r$ domination core $Z\subseteq V(G)$ of $G$ of size at most $\fcore(r,\epsilon)\cdot k^{1+\epsilon}$.
\end{lemma} 

If we make the assumption that the source and target sets
$D_s$ and $D_t$ are of minimum size, we can work with 
the improved bounds of \Cref{lem:core2} instead of the 
polynomial bounds of \Cref{lem:findcore}. The final obstacle is 
to better control the sizes of the trees $T_{v_\kappa}$ that we 
add to the graph $G'$ in the construction of \Cref{thm:reconfig-ds}. 
For this, we need the following two lemmas. 

\begin{lemma}[Lemma 2.9 of~\cite{drange2016kernelization}, adjusted (see Lemma 8 of~\cite{eickmeyer2016neighborhood})]\label{lem:closure}
There exists a function $\fcl(r,\epsilon)$ and a polynomial-time algorithm that, given $G\in \CCC$, $X\subseteq V(G)$, $r\in \N$, and $\epsilon>0$, computes the {\em{$r$-closure}} of~$X$, denoted $\cl_r(X)$
with the following properties. 
\begin{itemize}
  \item $X\subseteq \cl_r(X)\subseteq V(G)$;
  \item $|\cl_r(X)|\leq \fcl(r,\epsilon)\cdot |X|^{1+\epsilon}$; and
  \item $|M_r^G(u,\cl_r(X))|\leq \fcl(r,\epsilon)\cdot |X|^{\epsilon}$ for each $u\in V(G)\setminus \cl_r(X)$.
\end{itemize}
\end{lemma}

We can now compute a breadth-first search tree with root $v_\kappa$ of depth at most
$r$ which stops whenever it first encounters a vertex of $\cl_r(Z)$. This gives us a tree $T_{v_\kappa}$ 
of size at most $\fcl(r,\epsilon)\cdot |Z|^\epsilon\cdot r$. However, 
as the breadth-first search does not
continue when meeting $\cl_r(Z)$, we now have to connect the vertices of $\cl_r(Z)$ with minimum length
paths (up to length $r$) to preserve all distances. This is possible as the next lemma shows. 

\begin{lemma}[Lemma 2.11  of~\cite{drange2016kernelization}, adjusted (see Lemma 9 of~\cite{eickmeyer2016neighborhood})]\label{lem:pathsclosure}
There is a function $\fpaths(r,\epsilon)$ and a polynomial-time algorithm which on input $G\in \CCC$, $X\subseteq V(G)$, $r\in \N$, and $\epsilon>0$, 
computes a superset $X'\supseteq X$ of vertices with the following properties:
\begin{itemize}
\item whenever $\dist_G(u,v)\leq r$ for $u,v\in X$, then $\dist_{G[X']}(u,v)=\dist_G(u,v)$; and
\item $|X'|\leq \fpaths(r,\epsilon)\cdot |X|^{1+\epsilon}$.
\end{itemize}
\end{lemma}

We can now prove the following theorem. 

\begin{theorem}\label{thm:reconfig-ds2}
Let $\CCC$ be a nowhere dense class of graphs and let $r\in \N$. 
Let $(G,D_s,D_t,k)$ be an instance of $r$-DSR, where $G\in \CCC$
and where $D_s$ and $D_t$ are minimum \mbox{distance-$r$} dominating
sets in~$G$. 
There is a function $\fker(r,\epsilon)$ and a polynomial-time algorithm
which on input $(G,D_s,D_t,k)$ computes a subgraph $G'\subseteq G$ with 
$D_s,D_k\subseteq G'$ and $Z\subseteq V(G')$ 
such that $(G,k,D_s,D_t)$ is a positive
instance of $r$-DSR if and only if $((G',Z),k,D_s,D_t)$ is a positive instance
of $r$-ZDSR and $G'$ has order at most $\fker(r,\epsilon)\cdot k^{1+\epsilon}$. 
\end{theorem}
\begin{proof}
The proof parallels that of \Cref{thm:reconfig-ds}. We compute a 
distance-$r$ domination core $Z\subseteq V(G)$ using 
\Cref{lem:core2} instead of \Cref{lem:findcore} (with parameter $\epsilon'$ which will 
be determined in the course of the proof). Now, using \Cref{lem:closure}, we compute \mbox{$Z'=\cl_r(Z)$} and using \Cref{lem:pathsclosure}
we compute $Z''\supseteq Z'$ such that 
whenever $\dist_G(u,v)\leq r$ for \mbox{$u,v\in Z'$}, then $\dist_{G[Z'']}(u,v)=\dist_G(u,v)$. 
We classify the elements of $V(G)\setminus Z'$
according to their $r$-projections to $Z'$, that is, we define
$u\cong_{Z',r} v\Leftrightarrow \rho_r^G[u,Z']=\rho_r^G[v,Z']$.

We now construct $G'$ such that it contains $D_s, D_t$, the 
set $Z''$, all the representatives~$v_\kappa$ and furthermore
a small set $T$ of vertices such that $N_r^G(v_\kappa)\cap Z'=N_r^{G'}(v_\kappa)\cap Z'$. 
The set~$T$ is constructed as follows. For each $v\in V(G)$, let $T_v$
be a breadth-first search tree with root~$v$ of depth $r$ which does not continue
when meeting $Z$ for the first time. The crucial claim is the following. 

\setcounter{theorem}{0}
\medskip
\begin{claim}
Let $v_\kappa$ be a representative vertex. Then $\dist_G(v_\kappa,z)=\dist_{G'}(v_\kappa,z)$
for all $z\in N_r^G(v_\kappa)\cap Z$.
\end{claim}
\begin{proof}\renewcommand{\qedsymbol}{$\lrcorner$}
Let $z\in N_r^G(v_\kappa)\cap Z$ and let $P$ be a minimum length path between $v_\kappa$ 
and $z$. Let $z'$ be the first vertex on $P$ which belongs to $Z'$. Then we have 
$\dist_G(v_\kappa,z')=\dist_{T_{v_\kappa}}(v_\kappa,z')$. Now by construction 
of $Z''$ we have $\dist_G(z',z)=\dist_{G'}(z',z)$, which implies the claim.
 \end{proof}

\medskip
The rest of the proof works exactly as the proof of \Cref{thm:reconfig-ds}. Let
us determine a bound on the size of $G'$, which also determines our initial choice of
$\epsilon'$. The set $Z$ has size at most $\fcore(r,\epsilon')\cdot k^{1+\epsilon'}$. 
According to \Cref{lem:projection-complexity}
there are at most $\fproj(r,\epsilon')\cdot |Z|^{1+\epsilon'}$ projection 
classes, hence we add at most so many vertices $v_\kappa$ to 
$G'$. The set $Z'$ has size at most $\fproj(r,\epsilon')\cdot |Z|^{1+\epsilon'}$
according to \Cref{lem:closure} and the set $Z''$ has size at most 
$\fpaths(r,\epsilon')\cdot |Z'|^{1+\epsilon}$ according to 
\Cref{lem:pathsclosure}. Now, each tree $T_v$ has order 
at most $|Z'|^{\epsilon'}\cdot r$. Hence in total we have
\begin{align*}
V(G')| & \leq |Z''|+(\fproj(r,\epsilon')\cdot |Z|^{1+\epsilon'}+2)\cdot |Z'|^{\epsilon'}\cdot r\\
& \eqqcolon \fker(r,\epsilon)\cdot k^{1+\epsilon} 
\end{align*}
for an appropriately chosen function $f_{\mathrm{ker}}$ and $\epsilon'$.
\end{proof}

Observe that the constructed kernel preserves shortest reconfiguration
sequences. 

\subsection{Lower bounds}

\setcounter{theorem}{14}
\begin{theorem}
Let $\CCC$ be somewhere dense and closed under taking
subgraphs. Then there is $r\in \N$ such that 
\textsc{Distance-$r$ Dominating Set 
Reconfiguration} is $\mathsf{W}[2]$-hard on~$\CCC$. 
\end{theorem}
\begin{proof}
The proof works in principle as the proof of 
\Cref{thm:isw1}. Again, let $s\in \N$ be the number such that according to 
\Cref{lem:subdiv} for all graphs $G$ we have $G_s\in\CCC$. 
We reduce $1$-\textsc{DSR} to $(3s)$-\textsc{DSR} on $\CCC$ 
by finding an appropriate subdivision of a graph in which dominating
sets are translated $1$-to-$1$ to distance-$(3s)$ dominating sets. 
Here we can directly use the reduction from set cover to 
distance-$r$ dominating set from \cite{drange2016kernelization}, where we use the fact that
dominating set and set cover are equivalent problems (just define 
the set system consisting of the neighborhoods of all vertices). 
Now use that the reconfiguration variant of dominating
set is $\mathsf{W}[2]$-hard \cite{mouawad2017parameterized}.
\end{proof}

\section{Conclusion}

The study of computationally hard problems on restricted classes
of inputs is a very fruitful line of research in algorithmic graph structure
theory and in particular in parameterized complexity theory. This
research is based on the observation that many problems such as
\textsc{Dominating Set}, which are considered intractable in general,
can be solved efficiently on restricted graph classes. Of course it 
is a very desirable goal in this line of research to identify the most
general classes of graphs on which certain problems 
can be solved efficiently. In this work we were able to identify
the exact limit of tractability for the reconfiguration variants of 
the distance-$r$ independent set problem and distance-$r$ dominating 
set problem on subgraph closed graph classes. Clearly, the main open
question is to identify the most general graph classes which are not subgraph closed
on which these problems admit efficient solutions. 

\pagebreak
\bibliographystyle{plainurl}

\begin{thebibliography}{10}

\bibitem{bonamy2013recoloring}
Marthe Bonamy and Nicolas Bousquet.
\newblock Recoloring bounded treewidth graphs.
\newblock {\em Electronic Notes in Discrete Mathematics}, 44:257--262, 2013.

\bibitem{bonsma2013complexity}
Paul Bonsma.
\newblock The complexity of rerouting shortest paths.
\newblock {\em Theoretical computer science}, 510:1--12, 2013.

\bibitem{bonsma2009finding}
Paul Bonsma and Luis Cereceda.
\newblock Finding paths between graph colourings: {PSPACE}-completeness and
  superpolynomial distances.
\newblock {\em Theoretical Computer Science}, 410(50):5215--5226, 2009.

\bibitem{bonsma2014complexity}
Paul Bonsma, Amer~E Mouawad, Naomi Nishimura, and Venkatesh Raman.
\newblock The complexity of bounded length graph recoloring and csp
  reconfiguration.
\newblock In {\em International Symposium on Parameterized and Exact
  Computation}, pages 110--121. Springer, 2014.

\bibitem{bousquet2017token}
Nicolas Bousquet, Arnaud Mary, and Aline Parreau.
\newblock Token jumping in minor-closed classes.
\newblock {\em 21st International Symposium on Fundamentals of Computation Theory, {FCT}}, pages 136--149, 2017.

\bibitem{cereceda2008connectedness}
Luis Cereceda, Jan Van Den~Heuvel, and Matthew Johnson.
\newblock Connectedness of the graph of vertex-colourings.
\newblock {\em Discrete Mathematics}, 308(5):913--919, 2008.

\bibitem{cereceda2009mixing}
Luis Cereceda, Jan Van~den Heuvel, and Matthew Johnson.
\newblock Mixing 3-colourings in bipartite graphs.
\newblock {\em European Journal of Combinatorics}, 30(7):1593--1606, 2009.

\bibitem{cereceda2011finding}
Luis Cereceda, Jan Van Den~Heuvel, and Matthew Johnson.
\newblock Finding paths between 3-colorings.
\newblock {\em Journal of graph theory}, 67(1):69--82, 2011.

\bibitem{cygan2015parameterized}
Marek Cygan, Fedor~V Fomin, {\L}ukasz Kowalik, Daniel Lokshtanov, D{\'a}niel
  Marx, Marcin Pilipczuk, Micha{\l} Pilipczuk, and Saket Saurabh.
\newblock {\em Parameterized algorithms}, Volume~3.
\newblock Springer, 2015.

\bibitem{DawarK09}
Anuj Dawar and Stephan Kreutzer.
\newblock Domination problems in nowhere-dense classes.
\newblock In {\em {IARCS} Annual Conference on Foundations of Software Technology and
               Theoretical Computer Science, {FSTTCS}} 2009, pages 157--168.
  Schloss Dagstuhl --- Leibniz-Zentrum f\"ur Informatik, 2009.

\bibitem{diestel2012graph}
Reinhard Diestel.
\newblock {\em Graph Theory, 4th Edition}, Volume 173 of {\em Graduate {T}exts
  in {M}athematics}.
\newblock Springer, 2012.

\bibitem{downey2013fundamentals}
Rodney~G Downey and Michael~R Fellows.
\newblock {\em Fundamentals of parameterized complexity}, Volume~4.
\newblock Springer, 2013.

\bibitem{drange2016kernelization}
P{\aa}l~Gr{\o}n{\aa}s Drange, Markus~Sortland Dregi, Fedor~V. Fomin, Stephan
  Kreutzer, Daniel Lokshtanov, Marcin Pilipczuk, Micha\l{} Pilipczuk, Felix
  Reidl, Fernando {S{\'{a}}nchez Villaamil}, Saket Saurabh, Sebastian Siebertz,
  and Somnath Sikdar.
\newblock Kernelization and sparseness: the case of {D}ominating {S}et.
\newblock In {\em 33rd Symposium on Theoretical Aspects of Computer Science, {STACS}}, pages 31:1--31:14.
  Schloss Dagstuhl --- Leibniz-Zentrum f\"ur Informatik, 2016.

\bibitem{eickmeyer2016neighborhood}
Kord Eickmeyer, Archontia~C. Giannopoulou, Stephan Kreutzer, O{-}joung Kwon,
  Michal Pilipczuk, Roman Rabinovich, and Sebastian Siebertz.
\newblock Neighborhood complexity and kernelization for nowhere dense classes
  of graphs.
\newblock In {\em 44th International Colloquium on Automata, Languages, and
  Programming, {ICALP} 2017}, pages 63:1--63:14, 2017.

\bibitem{flum2006parameterized}
J{\"o}rg Flum and Martin Grohe.
\newblock {\em Parameterized complexity theory}.
\newblock Springer Science \& Business Media, 2006.

\bibitem{gopalan2009connectivity}
Parikshit Gopalan, Phokion~G Kolaitis, Elitza Maneva, and Christos~H
  Papadimitriou.
\newblock The connectivity of boolean satisfiability: computational and
  structural dichotomies.
\newblock {\em SIAM Journal on Computing}, 38(6):2330--2355, 2009.

\bibitem{grohe2014deciding}
Martin Grohe, Stephan Kreutzer, and Sebastian Siebertz.
\newblock Deciding first-order properties of nowhere dense graphs.
\newblock In {\em Journal of the {ACM}}, 64(3):17:1--17:32, 2017.

\bibitem{haas2014k}
Ruth Haas and Karen Seyffarth.
\newblock The k-dominating graph.
\newblock {\em Graphs and Combinatorics}, 30(3):609--617, 2014.

\bibitem{haddadan2015complexity}
Arash Haddadan, Takehiro Ito, Amer~E Mouawad, Naomi Nishimura, Hirotaka Ono,
  Akira Suzuki, and Youcef Tebbal.
\newblock The complexity of dominating set reconfiguration.
\newblock In {\em Workshop on Algorithms and Data Structures}, pages 398--409.
  Springer, 2015.

\bibitem{hearn2005pspace}
Robert~A Hearn and Erik~D Demaine.
\newblock Pspace-completeness of sliding-block puzzles and other problems
  through the nondeterministic constraint logic model of computation.
\newblock {\em Theoretical Computer Science}, 343(1-2):72--96, 2005.

\bibitem{ito2011approximability}
Takehiro Ito and Erik~D Demaine.
\newblock Approximability of the subset sum reconfiguration problem.
\newblock In {\em International Conference on Theory and Applications of Models
  of Computation}, pages 58--69. Springer, 2011.

\bibitem{ito2011complexity}
Takehiro Ito, Erik~D Demaine, Nicholas~JA Harvey, Christos~H Papadimitriou,
  Martha Sideri, Ryuhei Uehara, and Yushi Uno.
\newblock On the complexity of reconfiguration problems.
\newblock {\em Theoretical Computer Science}, 412(12-14):1054--1065, 2011.

\bibitem{ito2012reconfiguration}
Takehiro Ito, Marcin Kami{\'n}ski, and Erik~D Demaine.
\newblock Reconfiguration of list edge-colorings in a graph.
\newblock {\em Discrete Applied Mathematics}, 160(15):2199--2207, 2012.

\bibitem{ito2014fixed}
Takehiro Ito, Marcin Kami{\'n}ski, and Hirotaka Ono.
\newblock Fixed-parameter tractability of token jumping on planar graphs.
\newblock In {\em International Symposium on Algorithms and Computation}, pages
  208--219. Springer, 2014.

\bibitem{ito2014parameterized}
Takehiro Ito, Marcin Kami{\'n}ski, Hirotaka Ono, Akira Suzuki, Ryuhei Uehara,
  and Katsuhisa Yamanaka.
\newblock On the parameterized complexity for token jumping on graphs.
\newblock In {\em International Conference on Theory and Applications of Models
  of Computation}, pages 341--351. Springer, 2014.

\bibitem{kaminski2011shortest}
Marcin Kami{\'n}ski, Paul Medvedev, and Martin Milani{\v{c}}.
\newblock Shortest paths between shortest paths.
\newblock {\em Theoretical Computer Science}, 412(39):5205--5210, 2011.

\bibitem{kaminski2012complexity}
Marcin Kami{\'n}ski, Paul Medvedev, and Martin Milani{\v{c}}.
\newblock Complexity of independent set reconfigurability problems.
\newblock {\em Theoretical computer science}, 439:9--15, 2012.

\bibitem{siebertz2016polynomial}
Stephan Kreutzer, Roman Rabinovich, and Sebastian Siebertz.
\newblock Polynomial kernels and wideness properties of nowhere dense graph
  classes.
\newblock In {\em Proceedings of the Twenty-Eighth Annual {ACM-SIAM} Symposium on Discrete
               Algorithms, {SODA} 2017}, pages 1533--1545. {SIAM}, 2017.

\bibitem{lokshtanov2017complexity}
Daniel Lokshtanov and Amer~E Mouawad.
\newblock The complexity of independent set reconfiguration on bipartite
  graphs.
\newblock {\em Proceedings of the Twenty-Ninth Annual {ACM-SIAM} Symposium on Discrete
               Algorithms, {SODA} 2018}, pages 185--195, {SIAM}, 2018.

\bibitem{lokshtanov2015reconfiguration}
Daniel Lokshtanov, Amer~E Mouawad, Fahad Panolan, MS~Ramanujan, and Saket
  Saurabh.
\newblock Reconfiguration on sparse graphs.
\newblock In {\em Workshop on Algorithms and Data Structures}, pages 506--517.
  Springer, 2015.

\bibitem{mouawad2015reconfiguration}
Amer Mouawad.
\newblock On reconfiguration problems: Structure and tractability.
\newblock 2015.

\bibitem{mouawad2015shortest}
Amer~E Mouawad, Naomi Nishimura, Vinayak Pathak, and Venkatesh Raman.
\newblock Shortest reconfiguration paths in the solution space of boolean
  formulas.
\newblock In {\em International Colloquium on Automata, Languages, and
  Programming}, pages 985--996. Springer, 2015.

\bibitem{mouawad2014vertex}
Amer~E Mouawad, Naomi Nishimura, and Venkatesh Raman.
\newblock Vertex cover reconfiguration and beyond.
\newblock In {\em International Symposium on Algorithms and Computation}, pages
  452--463. Springer, 2014.

\bibitem{mouawad2017parameterized}
Amer~E Mouawad, Naomi Nishimura, Venkatesh Raman, Narges Simjour, and Akira
  Suzuki.
\newblock On the parameterized complexity of reconfiguration problems.
\newblock {\em Algorithmica}, 78(1):274--297, 2017.

\bibitem{mouawad2014reconfiguration}
Amer~E Mouawad, Naomi Nishimura, Venkatesh Raman, and Marcin Wrochna.
\newblock Reconfiguration over tree decompositions.
\newblock In {\em International Symposium on Parameterized and Exact
  Computation}, pages 246--257. Springer, 2014.

\bibitem{nevsetvril2010first}
Jaroslav Ne{\v{s}}et{\v{r}}il and Patrice Ossona~de Mendez.
\newblock First order properties on nowhere dense structures.
\newblock {\em The Journal of Symbolic Logic}, 75(03):868--887, 2010.

\bibitem{nevsetvril2011nowhere}
Jaroslav Ne{\v{s}}et{\v{r}}il and Patrice Ossona~de Mendez.
\newblock On nowhere dense graphs.
\newblock {\em European Journal of Combinatorics}, 32(4):600--617, 2011.

\bibitem{nishimura2017introduction}
Naomi Nishimura.
\newblock Introduction to reconfiguration.
\newblock {\em Algorithms 11(4):52}, 2018.

\bibitem{pilipczuk2017wideness}
Micha{\l} Pilipczuk, Sebastian Siebertz, and Szymon Toru{\'n}czyk.
\newblock On the number of types in sparse graphs.
\newblock {\em Proceedings of the 33rd Annual {ACM/IEEE} Symposium on Logic in Computer
               Science, {LICS} 2018}, pages 799-808, {ACM}, 2017.

\bibitem{schwerdtfeger2013computational}
Konrad~W Schwerdtfeger.
\newblock A computational trichotomy for connectivity of boolean
  satisfiability.
\newblock {\em {JSAT}}, 8(3/4):173--195, 2014. 

\bibitem{van2013complexity}
Jan van~den Heuvel.
\newblock The complexity of change.
\newblock {\em Surveys in combinatorics}, 409:127--160, 2013.

\bibitem{wrochna2014reconfiguration}
Marcin Wrochna.
\newblock Reconfiguration in bounded bandwidth and treedepth.
\newblock {\em arXiv preprint arXiv:1405.0847}, 2014.

\end{thebibliography}

\end{document}